\documentclass[11pt]{article}
\usepackage{times}
\usepackage{graphics}
\usepackage{latexsym}
\usepackage{mathtools}
\usepackage{amsthm, amsfonts, amsmath, amssymb}
\usepackage{hyperref}
\usepackage{url}
\usepackage[subpreambles=true]{standalone}
\usepackage[utf8]{inputenc}          
\usepackage[T1]{fontenc}   
\usepackage{comment}
\usepackage{todonotes}
\usepackage{mathtools}
\usepackage{setspace}
\usepackage{graphicx}
\usepackage{subcaption}
\usepackage{mathrsfs} 

\usepackage{algorithm}
\usepackage[noend]{algpseudocode}

\newcommand{\DS}{\textsc{DS}}
\newcommand{\CDS}{\textsc{CDS}}
\newcommand{\IDS}{\textsc{IDS}}
\newcommand{\EDS}{\textsc{EDS}}
\newcommand{\RD}{\textsc{RD}}
\newcommand{\IRD}{\textsc{IRD}}
\newcommand{\PRD}{\textsc{PRD}}
\newcommand{\NP}{\textsf{NP}}
\newcommand{\FPT}{\textsf{FPT}}
\newcommand{\FPSUB}{\textsf{FPSUB}}
\newcommand{\WO}{\textsf{W[1]}}
\newcommand{\WT}{\textsf{W[2]}}

\newtheorem{lemma}{Lemma}
\newtheorem{theorem}{Theorem}
\newtheorem{corollary}{Corollary}
\newtheorem{construction}{Construction}

\usepackage{tikz} 
\usetikzlibrary{decorations.pathreplacing, decorations.markings} 
\usetikzlibrary{calc} 
\usepackage{xcolor}
    
\def\empty{} 
\pgfkeys{utils/.cd,
    set if not empty/.code n args={3}{%
        \def\arg{#2}%
        \ifx\arg\empty%
        \pgfkeys{#1={#3}}%
        \else%
        \pgfkeys{#1={#2}}%
        \fi%
    },
    set if labelpos not empty/.code n args={2}{%
        \def\arg{#1}%
        \ifx\arg\empty%
    \else%
        \def\argo{#2}%
        \ifx\argo\empty%
        \pgfkeys{/tikz/label={#1}}%
        \else%
        \pgfkeys{/tikz/label={#2}:{#1}}%
        \fi%
        \fi%
    },
    set if arrowpos not empty/.code n args={3}{%
        \def\arg{#1}%
        \def\arstyle{#3}%
        \ifx\arstyle\empty%
        \def\arr{\arrow[>=stealth]{>}}
        \else%
        \def\arr{#3}
        \fi%
        \ifx\arg\empty%
        \pgfkeys{/pgf/decoration/mark={at position #2 with {\arr}}}%
        \else%
        \pgfkeys{/pgf/decoration/mark={at position #1 with {\arr}}}%
        \fi%
    },
}
\tikzset{
    nodes/.style n args={4}{
        draw ,circle,outer sep=0.5mm,
        /utils/set if not empty={/tikz/fill}{#1}{black},
        /utils/set if not empty={/tikz/minimum size}{#4}{5},
        /utils/set if labelpos not empty={#2}{#3},
        line width = 0.5pt,
},
    arc/.style n args={3}{
        postaction={
            decorate,
            decoration={markings,
                /utils/set if arrowpos not empty={#1}{1}{}%
            }
        },
        /utils/set if not empty={/tikz/line width}{#2}{0.7pt},
        {#3}
    }
}
\tikzset{shorten >= 1pt}
\definecolor{Neutral}{RGB}{0,0,0}
\definecolor{Gray}{RGB}{123,123,123}
\definecolor{Blue}{RGB}{86,180,233}
\definecolor{Red}{RGB}{221,0,0}
\definecolor{Green}{RGB}{8,212,84}
\definecolor{Pink}{RGB}{236,149,198}
\definecolor{Purple}{RGB}{171,56,210}
\definecolor{Yellow}{RGB}{243,212,17}
\definecolor{Orange}{RGB}{230,103,0}
\definecolor{Brown}{RGB}{140,63,4}

\newcommand{\prob}[3]
{
\begin{quote}
  {\sc #1}
  \nopagebreak
  \vspace{-0.1cm}
  \begin{description}
  \item[\textbf{Input:}] #2
  \item[\textbf{Question:}] #3
  \end{description}
\end{quote}
}

\begin{document}

\title{Domination in Diameter-Two Graphs and the 2-Club Cluster Vertex Deletion Parameter}
\author{Faisal N. Abu-Khzam and Lucas Isenmann\\
Department of Computer Science and Mathematics\\
Lebanese American University\\
Beirut, Lebanon}

\date{}
\maketitle

\thispagestyle{empty}



\begin{abstract}

The $s$-club cluster vertex deletion number of a graph, or sccvd, is the minimum number of vertices whose deletion results in a disjoint union of $s$-clubs, or graphs whose diameter is bounded above by $s$. 
We launch a study of several domination problems on diameter-two graphs, or 2-clubs, and study their parameterized complexity with respect to the 2ccvd number as main parameter. 
We further propose to explore the class of problems that become solvable in sub-exponential time when the running time is independent of some input parameter. Hardness of problems for this class depends on the Exponential-Time Hypothesis. We give examples of problems that are in the proposed class and problems that are hard for it.

\end{abstract}

\section{Introduction}

The parameterized complexity of a problem focuses mainly on the asymptotic running time with respect to specific input parameters, which can provide insights into how the running time can be improved by selecting appropriate parameter(s). This type of algorithmic analysis reveals the potential for more efficient algorithms tailored to instances where some (chosen) parameters are small, thus enhancing the overall tractability of the problem.

Studying a problem with respect to a variety of possible input parameters has witnessed great increase in attention and there has been several studies in the recent literature \cite{Banik2023,  FELLOWS2013541,Goyal-Raman2024,Bart2014}. 
The relation between the various parameters is of interest by itself. For example, the deletion into bounded degree parameter (dbd) is more general than the vertex cover parameter (vc), thus a problem that is fixed-parameter tractable (\FPT) with respect to dbd is also \FPT\ with respect to vc. Moreover, if a problem does not admit a polynomial size kernel (modulo some hierarchy condition) with respect to vc, then it is hard to ``kernelize'' with respect to dbd. In general, given a problem $X$ that can be parameterized by any of two parameters $k$ and $t$, we say that parameter $k$ is smaller than parameter $t$ if and only if the family of Yes-instances for $t$ is a subset of the family of Yes-instances for $k$. As such, it would be more interesting from a complexity standpoint to study the complexity with respect to the smaller parameter.

There are many generalizations of the {\sc Vertex Cover} problem that have resulted in smaller parameters. A notable example is the cluster vertex deletion parameter (cvd): number of vertices whose deletion yields a disjoint union of cliques (of course, {\sc Vertex Cover} corresponds to the special case where each of the resulting cliques is of size one). Recently, the parameterized complexity of a variety of domination problems was studied with respect to cvd \cite{Goyal-Raman2024}. In this paper, we further refine the cluster deletion parameter by considering the $s$-club vertex deletion parameter (sccvd): number of vertices whose deletion results in a disjoint union of $s$-clubs, or graphs of diameter at most $s$. In particular, we consider the 2-club cluster vertex deletion parameter, or 2ccvd, being the very next parameter in the considered chain of ``deletion into $s$-clubs'' parameters for $s = 1,2 \ldots$

We mainly consider domination problems by studying several variants of the classic dominating set problem, namely: {\sc Efficient Dominating Set}, {\sc Independent Dominating Set} and a variety of {\sc Roman Domination} problems. To study parameterization with respect to 2ccvd, we first consider the complexity of each of these problems on graphs of diameter (at most) two. And then obtain the resulting parameterized complexity of the problem with respect to 2ccvd. 
Finally, we propose a study that can further classify problems from a parameterized complexity viewpoint, by considering problems that admit sub-exponential time algorithms when some input parameter(s) are small enough, or constant.

\section{Preliminaries}

Basic graph theoretic terminology is used throughout this paper. All the considered graphs are assumed to be undirected, simple and unweighted. The neighborhood of a vertex $v$ of a graph $G=(V,E)$ is the set of vertices adjacent to $v$ in $G$: $N(v)=N^1(v) =\{w\in V: uw\in E\}$, and its closed neighborhood is $N[v]=N^1[v] = N(v)\cup \{v\}$. 
This notion can be extended to sets by defining, for $S\subset V, N(S) = N^1(S) = \bigcup_{v\in S} N(v)$ and $N[S] = N^1[S] = \bigcup_{v\in S} N[v]$. Furthermore, for $j > 1$, $N^j[v] = N[N^{j-1}(v)]$ and $N^j(v) = N^j[v] \setminus \{v\}$. In other words, $N^j[v]$ is the set of vertices that are within distance $j$ from $v$. 

A graph $G=(V,E)$ is an $s$-club if for every vertex $v\in V, N^s[v] = V$, i.e., every pair of vertices of $G$ are within a distance of $s$ from each other. The subgraph induced by a set $S\subset V$ will be denoted by $G[S]$.
A dominating set in $G=(V,E)$ is a set $D\subset V$ such that $N[D]=V$. 
A set $D$ is said to be an {\em efficient dominating set} (EDS) if for every $v\in V, |N[v]\cap D| = 1$. 
It is an {\em independent dominating set} (IDS) if 
for every $v\in D, |N[v]\cap D| = 1$ (while $N[D]=V$). 
A dominating set $D$ is a connected dominating set of $G$ if $G[D]$ is connected. The {\sc $s$-Club Cluster Vertex Deletion} problem (sCCVD) is defined as follows:

\prob{$s$-Club Cluster Vertex Deletion}
{A graph $G=(V,E)$ and a non-negative integer $k$;}{Is there a set $S\subset V$ of cardinality at most $k$ such that $G-S$ is a disjoint union of $s$-clubs?}

The sCCVD problem seems to have been first formulated in \cite{SchaeferPhdThesis} where it was shown to be solvable in polynomial-time on trees. The 2CCVD problem was studied in \cite{Liu} where it was shown to be \NP-hard but fixed-parameter tractable by presenting an $\mathcal{O}^*(3.31^k)$ algorithm. The same was proved in \cite{Liu} for the edge-deletion version. 
Further improvements and similar problems were obtained/discussed in a number of recent papers \cite{abu2023, Niedermeier2021, Tsur2024}.

We define the $sccvd$ of a graph $G$ as the minimum integer $k$ for which sCCVD has a solution in $G$. We assume $s$ is a small constant so computing $sccvd(G)$ takes \FPT-time. In fact, the sCCVD problem is trivially \FPT\ when $s$ is a constant being solvable in $\mathcal{O}((s+1)^k)$ time: find a path $p$ of length $s$ between two vertices that are at distance exactly $s$ from each other; then branch by deleting one of the $s+1$ elements of $p$ in each case.

As mentioned in the introduction, we focus on the 2ccvd parameter in this paper, and we shall first consider variants of the {\sc Dominating Set} problem. 
We first note that the {\sc Dominating Set} (DS) and {\sc Connected Dominating Set} (CDS) problems are \textsf{para-NP}-hard with respect to sccvd, simply because these problems are \WT-hard in $2$-clubs \cite{lokshtanov2013hardness}. 
Therefore we focus on other variants of {\sc Dominating Set} that have not been studied on diameter-2 graphs. In particular, we consider {\sc Efficient Dominating Set} (EDS), {\sc Independent Dominating Set} (IDS), {\sc Roman Domination} (RD), {\sc Independent Roman Domination} (IRD) and {\sc Perfect Roman Domination} (PRD). A summary of the known and proved results is given in the below table.

\begin{table}[htb!]
    \centering
    \begin{tabular}{c|c|c}
        Problem(s) & On Diameter-$2$ graphs & Parameterized by 2ccvd \\
        \hline
        \DS, \CDS & \WT-hard \cite{lokshtanov2013hardness} & \textsf{para-NP}-hard  \\
        \IDS, \RD, \IRD & \WO-hard  (Theorems \ref{theorem:IDS_W1H_2clubs}, \ref{theorem:RD_W1} \& \ref{theorem:IRD_2clubs})  & \textsf{para-NP}-hard \\
        \EDS & Linear-time  & \FPT (Theorem \ref{theorem:EDS_FPT}) \\
        \PRD & \NP-hard (Theorem \ref{theorem:PRD}) & \textsf{para-NP}-hard
    \end{tabular}
    \caption{Summary of the complexities of domination problems parameterized by 2ccvd.}
    \label{tab:summary}
\end{table}

\section{{\sc Efficient Dominating Set} parameterized by 2ccvd}

We prove that {\sc Efficient Dominating Set} parameterized by 2ccvd is \FPT\ 
by exhibiting a fixed parameter tractable algorithm.

Let $G$ be a graph such that $2ccvd(G)\leq k$.
Then there exists a subset $S \subseteq V(G)$ of size $k$ or less such that $G - S$ is a disjoint union of 2-clubs that we denote by $C_1, \ldots, C_p$.
Let $S' \subseteq S$  be such that $S'$ is an EDS of $G[N[S']]$ (meaning that the sets $N[s]$ are disjoint for every $s \in S'$).
We define $U = S \setminus N[S']$ the vertices of $S$ not dominated by $S'$.
If there exists $i \in [1,p]$ such that $V(C_i) \subseteq N[S']$, then we can remove $C_i$.
Therefore we suppose that for every $i \in [1,p]$, we have $V(C_i) \not\subseteq N[S']$.

\begin{lemma}\label{lemma:partial_solution_extension}
    A subset $D$ of $V(G)$ is an EDS of $G$ such that $S' = D \cap S$ if and only if there exists $(v_1, \ldots, v_p) \in V(C_1) \times \cdots \times V(C_p)$ such that
    \begin{itemize}
        \item $D = S' \sqcup_{i \in [1,p]} \{v_i\}$;
        \item $U = \sqcup_{i \in [1,p]} (N[v_i] \cap U)$;
        \item $\forall i \in [1,p]$,  $V(C_i) \subseteq N[v_i] \cup N[S']$ and $N[v_i] \cap N[S'] = \emptyset$.
    \end{itemize}
\end{lemma}
\begin{proof}
    Suppose that $D$ is an EDS of $G$ such that $S' = D \cap S$.

    Let $i \in [1,p]$.
    Remark that for any two vertices $v$ and $w$ in $C_i$ we have $N[v] \cap N[w] \not=\emptyset$ because $d(v,w) \leq 2$.
    Therefore $|D \cap V(C_i)| \leq 1$.
    If $D \cap V(C_i) = \emptyset$, then we would have $V(C_i) \subseteq N[S']$ which is excluded.
    Therefore there exists $v_i \in V(C_i)$ such that $D \cap V(C_i) = \{v_i\}$.
    We deduce that $D = S' \sqcup_{i \in [1,p]} \{v_i\}$.
    As $D$ is an EDS, then the sets $N[v_i]$ are disjoint.
    As $U \cap N[S']$ is by definition empty, then $U$ must be covered by the neighbors of the vertices $v_1, \ldots, v_p$.
    Thus $U = \sqcup_{i \in [1,p]}( N[v_i] \cap U)$.

    Let $i \in [1,p]$.
    As $D$ is an EDS of $G$, then $N[v_i]$ and $N[S']$ are disjoint.
    Furthermore $V(C_i)$ can only be covered by the neighbors of $v_i$ or $S'$.
    Thus $ V(C_i) \subseteq N[v_i] \cup N[S']$.
    \newline
    
    Now suppose there exists $(v_1, \ldots, v_p) \in V(C_1) \times \cdots \times V(C_p)$ such that the previous properties are satisfied.
    Let us prove that $D$ is an EDS of $G$.
    Consider $i$ and $j$ in $[1,p]$.
    Then $N[v_i]$ and $N[v_j]$ does not intersect in the clusters $C_1, \ldots, C_p$ and in $U$.
    Furthermore they do not intersect $N[S']$.
    Thus $N[v_i]$ and $N[v_j]$ are disjoint.
    Therefore the sets $N[v_1], \ldots, N[v_j]$ and $N[S']$ are pairwise disjoint.
    Furthermore for every $i \in [1,p]$, $V(C_i)$ is covered by $v_i$ and $S'$.
    The set $U$ is covered by $v_1, \ldots, v_p$.
    We conclude that $D$ is an EDS. 
\end{proof}

We present a dynamic programming method for computing a minimum EDS solution, if it exists.
\newline

For every subset $W \subseteq U$ and $j \in [1,p]$ we define $T[W,j]$ as a Boolean which is true if and only if there exists $(v_1, \ldots, v_j) \in V(C_1) \times \cdots \times V(C_j)$ such that 
\begin{itemize}
    \item $W = \sqcup_{i \in [1,j]} (N[v_i] \cap U)$:
    \item $\forall i \in [1,j]$, $V(C_i) \subseteq N[v_i] \cup N[S']$ and $N[v_i] \cap N[S'] = \emptyset$.
\end{itemize}

Because of Lemma~\ref{lemma:partial_solution_extension}, there exists an EDS, denoted $E$, of $G$ such that $S' =E \cap S$ if and only if $T[U, p]$ is true.

\begin{lemma}
    For every $W \subseteq U$ and $2 \leq j \in [1,p]$,
    $$ T[W, j] = \lor_{v \in V(C_j)_{S'}} T[W - N[v], j-1] $$
    where $V(C_j)_{S'} = \{ v \in V(C_j) \mid V(C_j) \subseteq N[v_j] \cup N[S'] \land N[v_j] \cap N[S'] = \emptyset \}$.
\end{lemma}
\begin{proof}
    Let $W \subseteq U$ and $2 \leq j \in [1,p]$.
    Suppose that there exists $v_j \in V(C_j)_{S'}$ such that $T[W-N[v_j], j-1]$ is true.
    Then there exists $(v_1, \ldots, v_{j-1})$ in $V(C_1) \times \cdots \times V(C_{j-1})$ such that 
    \begin{itemize}
        \item $W-N[v_j] = \sqcup_{i \in [j-1]} (N[v_i] \cap U)$;
        \item $\forall i \in [j-1]$,  $V(C_i) \subseteq N[v_i] \cup N[S']$ and $N[v_i] \cap N[S'] = \emptyset$.
    \end{itemize}  
    
    Let us prove that $W = \sqcup_{i \in [1,j]} (N[v_i] \cap U)$.
    The subsets $N[v_i] \cap U$ are disjoint for any $i \in [j-1]$.
    Moreover, for every $i \in [j-1]$, $N[v_i] \cap U \subseteq W-N[v_j]$ thus $N[v_i] \cap U \cap N[v_j] = \emptyset$.
    Thus the subsets $N[v_i] \cap U$ are disjoint for any $i \in [j]$.
    Let $x \in W$.
    If $x \in W\setminus N[v_j]$, then $x \in  \sqcup_{i \in [j-1]} (N[v_i] \cap U)$ and then $x \in \sqcup_{i \in [j]} (N[v_i] \cap U)$.
    If $x \in N[v_j] \cap U$, then $x \in  \sqcup_{i \in [1,j]} (N[v_i] \cap U)$.
    By double inclusion, we have $W = \sqcup_{i \in [1,j]} (N[v_i] \cap U)$.
    Furthermore $\forall i \in [1,j]$, $V(C_i) \subseteq N[v_i] \cup N[S']$ and $N[v_i] \cap N[S'] = \emptyset$.
    We conclude that $T[W,j]$ is true.
    \newline

    Suppose that $T[W,j]$ is true.
    Then 
    $\exists (v_1, \ldots, v_j) \in V(C_1), \ldots , V(C_j)$ such that:
    \begin{itemize}
        \item $W = \sqcup_{i \in [1,j]} (N[v_i] \cap U)$;
        \item $\forall i \in [1,j]$, $N[v_i] \cap V(C_i) = V(C_i) - N[S']$ and $N[v_i] \cap N[S'] = \emptyset$.
    \end{itemize}
    
    Let us prove that $T[W-N[v_j], j-1]$ is true.
    In a similar way as before we prove that $W -N[v_j] = \sqcup_{i \in [j-1]} (N[v_i] \cap U) $.
    Furthermore remark that $\forall i \in [1,j-1]$, $v_i \in V(C_i)_{S'}$.
    We deduce that $T[W-N[v_j], j-1]$ is true.
    We conclude that
     $$ T[W, j] = \lor_{v \in V(C_i)_{S'} } T[W - N[v], j-1] \;. $$
\end{proof}

The following algorithm
tries to extend a subset $S'$ of $S$ to an EDS.

\begin{algorithm}
\caption{EDS-2CCVD fixed parameter tractable algorithm}
\begin{algorithmic}[1]
\Procedure{CanExtendPartialSolution}{$S'$}
    \State $U \gets S - (N[S'] \cap S)$
    \For{$W \subseteq U$}
        \For{$j \in [1,p]$}
            \State $T[W,j] \gets$ false
        \EndFor
    \EndFor
    \For{$v_1 \in V(C_1)$}
         \If{$V(C_1) \subseteq N[v_1] \cup N[S']$ and $N[v_1] \cap N[S'] = \emptyset$}
             \State $T[N[v_1],1] \gets$ true
        \EndIf
    \EndFor

    \For{$j \in [2,p]$}
        \For{$W \subseteq U$}
            \For{$v \in V(C_j) \mid V(C_j)  \subseteq N[v_j] \cup N[S']$ and $N[v] \cap N[S'] = \emptyset$}
                \If{$T[W-N[v], j-1]$}
                    \State $T[W,j] \gets $ true
                \EndIf
            \EndFor
        \EndFor 
    \EndFor
    \Return{T[U, p]}
\EndProcedure
\end{algorithmic}
\end{algorithm}

\begin{lemma}
The time complexity of Algorithm 1 is in $O(n^3 2^{|U|})$.
\end{lemma}

\begin{proof}
By precomputing in $O(n^3)$ for all $j \in [1,p]$ the sets $V(C_j) \setminus N[v_j]$ and for all $v \in V(G)$, $N[v] \cap N[S'] = \emptyset$, then we can check in $O(n^2)$ if a vertex $v \in V(C_j)$ satisfies $V(C_j) \subseteq N[v_j] \cup N[S']$ and $N[v_j] \cap N[S'] = \emptyset$.
Thus the complexity of \texttt{CanExtendPartialSolution} is $O( n^3 + n 2^{|U|} + n  2^{|U|} n^2)$ which is in $O(n^3 2^{|U|})$.
\end{proof}

We now consider the following algorithm, which tries to extend every EDS of $N[S']$ to $G$ and returns the size of a minimum EDS of $G$.

\begin{algorithm}
\caption{EDS}
\begin{algorithmic}[1]
\Procedure{EDS}{}
    \State $r \gets +\infty$
    \For{$S' \subseteq S$}
        \If{$S'$ is an EDS of $N[S']$}
            \State $k \gets$ the number of $2$-clubs not dominated by $S'$
            \If{ CanExtendPartialSolution($S'$) and $|S'| +k <r$}
                \State $r \gets |S'|+k$
            \EndIf
        \EndIf
    \EndFor
    \Return{r}
\EndProcedure
\end{algorithmic}
\end{algorithm}

\begin{lemma}
The time complexity of Algorithm 2 is in $O(n^3 3^k)$.
\end{lemma}
\begin{proof}
Given a subset $S' \subseteq S$, we can check in $O(k^2 n^2)$ if $S'$ is an EDS of $N[S']$ (equivalent to checking if all subsets $N[s]$ are pairwise disjoint for every $s \in S'$). As $|U| \leq k - |S'|$, we deduce that the complexity of the total algorithm is $O( n^3 3^k)$ because $\sum_{S' \subseteq S} 2^{-|S'|} = (3/2)^k$.
\end{proof}

To conclude:

\begin{theorem} \label{theorem:EDS_FPT}
    {\sc Efficient Dominating Set} parameterized by 2CCVD is \FPT.
\end{theorem}

We can give a lower bound for the complexity of {\sc Efficient Dominating Set} parameterized by 2ccvd thanks to the following Theorem:

\begin{theorem}[Goyal \textit{et al.} \cite{goyal2024parameterized}]
    {\sc Efficient Dominating Set} parameterized by {\sc Vertex Cover} cannot be solved in $2^{o(|S|)}$ time unless ETH fails.
\end{theorem}

Since any vertex cover is $s$-club cluster vertex deletion set, we deduce that:

\begin{corollary}
    {\sc Efficient Dominating Set} parameterized by $s$ccvd cannot be solved in $2^{o(|S|)}$ time unless ETH fails.
\end{corollary}

\section{Complexity of {\sc Independent Dominating Set} on $2$-clubs}

The problem {\sc Independent Dominating Set} is known to be \NP-complete \cite{garey1979computers}.
In order to study the complexity of the latter problem, we consider this variation of {\sc Independent Set}:

\prob{$k$-Multicolored Independent Set}
{A graph $G$ and a vertex coloring $c : V (G) \to \{1, 2, . . . , k\}$ for $G$;}
{Does G have an independent set including vertices of all $k$ colors? That is, are there $v_1, \ldots, v_k \in V(G)$ such that for all $1 \leq i < j \leq k$, $\{v_i, v_j\} \not\in E(G)$ and $c(v_i) \not= c(v_j)$?}

According to \cite{DBLP:journals/tcs/FellowsHRV09}, {\sc $k$-Multicolored Independent Set} is $W[1]$-complete.

\begin{theorem} \label{theorem:IDS_W1H_2clubs}
    {\sc Independent Dominating Set} is $W[1]$-hard on $2$-clubs.
\end{theorem}
\begin{proof}
    We reduce {\sc Independent Dominating Set} from {\sc $k$-Multicolored Independent Set}.
    Consider a graph $G$ whose vertices are colored with $k$ colors $1, \ldots, k$.
    For every $i \in [1,k]$, we denote by $C_i$ the subset of vertices colored $i$.
    We construct a graph $G'$ from $G$ as follows:
    \begin{itemize}
        \item Add edges between every pair of vertices having the same color (so that $C_i$ becomes a clique in $G'$ for every $i \in [1,k]$);
        \item For every $i \in [1,k]$, add two vertices $x_i$ and $y_i$ connected to all the vertices colored $i$;
        \item Add 4 vertices $vx$, $\overline{vx}$, $vy$ and $\overline{vy}$, and connect them to all the initial vertices;
        \item For every $i$, connect $vx$ and $\overline{vx}$ to $x_i$, and connect $vy$ and $\overline{vy}$ to $y_i$;
        \item For every $i \not= j$, add 2 vertices $x_i y_j$ and $\overline{x_i y_j}$ connected to $x_i$ and $y_j$;
        \item Connect the vertices $vx$, $\overline{vx}$, $vy, \overline{vy}$, $x_i y_j, \overline{x_i y_j}$ together except for each vertex and its overlined counterpart;
        \item Add a vertex $u$ connected to the vertices  $vx$, $\overline{vx}$, $vy, \overline{vy}$, $x_i y_j, \overline{x_i y_j}$.
    \end{itemize}

    Remark that with our notations, a vertex is the twin of its overlined counterpart, meaning they share the same neighbors.
    We call {\it central} vertices each of $vx$, $\overline{vx}$, $vy$, $\overline{vy}$, $x_i y_j$ and $\overline{x_i y_j}$.

   \begin{figure}[htb!]
   \centering
     \begin{subfigure}[t]{0.45\textwidth}
        \centering
         \includegraphics[width=\textwidth]{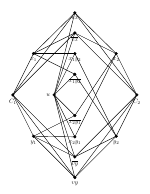}
         \caption{The graph $G'$ obtained from the union of two disjoint vertices colored respectively $1$ and $2$. The edges between the central vertices are not drawn.}
     \end{subfigure}
     \quad
     \begin{subfigure}[t]{0.45\textwidth}
        \centering
         \includegraphics[width=\textwidth]{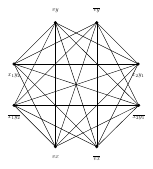}
         \caption{The induced subgraph of the central vertices.}
     \end{subfigure}
     \caption{Example of a graph obtained by the construction from an instance of $k$-{\sc Multicolored Independent Set}.}
\end{figure}

    We prove that $G'$ is a $2$-club by demonstrating in Table~\ref{tab:common_neighbors_reduction} a common neighbor between each pair of vertices.
    \begin{table}[htb!] 
        \centering
        \begin{tabular}{c|cccccccc}
             &  $C_i$ & $x_i$ & $y_i$ & $vx$ & $vy$ & $x_i y_j$ & $u$ \\ \hline
             $C_{i'}$ & $vx$ & $vx$ & $vy$ & $vx$ & $vy$ & $vx$ & $vx$ \\
             $x_{i'}$ &  & $vx$ & $x_{i'} y_j$ & $vx$ & $vx$ & $vx$ & $vx$\\
             $y_{i'}$ & & & $vy$ & $vy$ & $vy$ & $vy$ & $vy$ \\
             $vx$ &  & & & $u$ & $u$ & $u$ & $u$ \\
             $vy$ & & & & & $u$ & $u$ & $u$ \\
             $x_{i'} y_{j'}$ & & & & & & $u$ & $u$ \\
             $u$ & & & & & & & $u$
        \end{tabular}
        \caption{Common neighbors of $G'$ for each type of vertices, proving that $G'$ is of diameter $2$. Sub-diagonal cases are symmetric.}
        \label{tab:common_neighbors_reduction}
    \end{table}

    Given a colorful independent set $X$ of $G$, we define the following subset of $G'$ by $X' = X \cup \{u\}$.
    As $X$ is of size $k$, $X'$ is of size $k+1$.
    As $X$ is an independent set of $G$, then $X$ is still an independent set of $G'$.
    Furthermore $u$ is not connected to the original vertices.
    Therefore $X'$ is an independent set of $G'$.
    For every $i \in [1,k]$, every vertex of $V_i$ is dominated by $X_i$ (the vertex of $X$ colored $i$).
    Furthermore every central vertex is dominated by $u$.
    Thus $X'$ is an IDS of size $k+1$ of $G'$.

    Suppose that $X'$ is an IDS of size at most $k+1$ of $G'$.
    Suppose by contradiction that $u$ is not in $X'$.
    As $u$ should be dominated by $X'$, then there exists a central vertex $w$ vertex in $X'$.
    Its twin should also be in $X'$ (it is a general fact that twins are either both in an IDS either they are not in it).
    Let $i \in [1,k]$ and let us show that $X' \cap ( C_i \cup \{x_i, y_i\})$ is not empty.
    By construction the vertex $w$ cannot be adjacent to both $x_i$ and $y_i$ because the common neighbors of $x_i$ and $y_i$ are the vertices of $C_i$.
We can assume without loss of generality that 
$w$
is not adjacent to $x_i$.
    The vertex $x_i$ cannot be dominated by another central vertex because $w$ is already in $X'$ and is connected to every central vertex except $\overline{w}$.
    Thus $x_i$ must be dominated by itself or by a vertex of $C_i$.
    Thus $X' \cap ( C_i \cup \{x_i, y_i\})$ is not empty for every $i$.
    As the subsets $ C_i \cup \{x_i, y_i\}$ are disjoint and as these subsets are disjoint from the central vertices, then $|X'| \geq k + 2$ which contradicts that $X' \leq k+1$.
    We deduce that $u$ is in $X'$.

    Thus no central vertices can be in $X'$.
    For every $i \in [1,k]$, the vertices $x_i$ and $y_i$ must be dominated by themselves or by a vertex of $C_i$.
    If no vertex of $C_i$ is in $X'$, then $x_i$ and $y_i$ are in $X'$.
    By the same previous counting argument, it is impossible.
    Thus $X'$ contains a vertex from $C_i$ for every $i$.
    We deduce that there exists a colorful independent set in $G$.
\end{proof}

\begin{corollary}
{\sc Independent Domination Set} parameterized by 2ccvd is \textsf{para-NP}-hard.    
\end{corollary}

\section{Roman domination}

The notion of Roman domination was first introduced in \cite{stewart1999defend}. A {\it Roman dominating function} (RDF) over on a graph $G$ is a function $f: V(G) \to \{0,1,2\}$ such that for every vertex $v \in f^{-1}(\{0\})$, there exists $u \in N(v)$ such that $f(u) = 2$.
The weight $w(f)$ of such a function $f$ is defined as $\sum_{v \in V(G)} f(v)$.
There always exists such a function as the constant function $1$ is a RDF.
We denote by $\gamma_R(G)$ the minimum weight of a RDF in $G$. 
This notion is related to dominating sets because considering a dominated set $X$ in $G$, the function $f_X$ defined by $f_X(v) = 2$ if $v \in X$ and $f_X(v) = 0$ otherwise is a RDF.
Therefore $\gamma_R(G) \leq 2\gamma(G)$.
The corresponding
minimization problem is formally defined as follows.

\prob{$k$-Roman Domination}
{A graph $G$, an integer $k$;}
{Is there a RDF of weight at most $k$?}

Similarly we recall the notion of {\it independent Roman dominating function} (IRDF), introduced in \cite{cockayne2004roman}, over a graph $G=(V,E)$: It is a RDF such that $f^{-1}(\{1,2\})$ is an independent set.
For a survey on {\sc Independent Roman Domination} see \cite{padamutham2021complexity}.
We consider the associated mimization problem:

\prob{$k$-Independent Roman Domination}
{A graph $G$, an integer $k$;}
{Is there an IRDF of weight at most $k$?}

Another variant is the notion of {\it perfect Roman dominating function} (PRDF) introduced in \cite{henning2018perfect}, over a graph $G=(V,E)$: It is a RDF such that for every $v \in V(G)$ with $f(v) = 0$, there exists exactly one vertex $w \in N(v)$ with $f(w) = 2$.
The associated decision problem is \textsf{W[1]}-complete if parameterized by solution size and fixed parameter tractable if parameterized by clique-width \cite{mann2024perfect}.

\prob{$k$-Perfect Roman Domination}
{A graph $G$, an integer $k$;}
{Is there a PRDF of weight at most $k$?}






\subsection{Complexity of {\sc Roman domination} in 2-Clubs}





\begin{construction}
\label{construction:double_copy}
Let $G$ be an arbitrary graph. We construct the following graph $G'$.
Consider the union of two copies $G_1$ and $G_2$ of $G$.
For every $v \in G_1$, connect $v$ to $N[v']$ where $v'$ is the copy of $v$ in $G_2$. For every $v' \in G_2$, connect $v'$ to $N[v]$ in $G_1$. 

\noindent
Claim: If $G$ is a $2$-club then $G'$ is a $2$-club.
\end{construction}

\begin{proof}
Consider two vertices $v$ and $w$ of $G'$.
    If $v$ and $w$ are in the same copy of $G$, then there exists a common neighbor in this copy (as $G$ is a $2$-club).
    Otherwise, we can suppose without loss of generality that $v \in G_1$ and $w \in G_2$.
    We denote by $v'$ the copy of $v$ in $G_2$.
    As $G_2$ is a $2$-club, $v'$ and $w$ have a common neighbor that we call $u$.
    As $N[v] = N[v']$, then $u$ is also a neighbor of $v$.
    Thus $v$ and $w$ have a common neighbor. Therefore, $G'$ is a 2-club.
\end{proof}

\begin{figure}[h]
    \centering
    \includegraphics[width=0.4\textwidth]{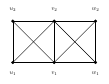}
    \caption{Example of the duplication construction: The depicted graph is the graph $G'$ obtained from a graph $G$ which is $3$-path with vertex set $\{u,v,w\}$.}
    \label{fig:IRD_reduction_example}
\end{figure}

\begin{theorem} \label{theorem:RD_W1}
    {\sc Roman Domination} is \textsf{W[1]}-hard even restricted to $2$-clubs.
\end{theorem}
\begin{proof}
We use a reduction 
from {\sc Dominating Set} in 2-clubs, which is $W[2]$-hard \cite{lokshtanov2013hardness}. Consider a 2-club $G$ and the $2$-club $G'$ obtained from Construction~\ref{construction:double_copy}. We prove that 
$G$ has a dominating set of size at most $k$
if and only if
$G'$ has a RDF of weight at most $2k$.

Suppose that $G$ has a dominating set $X$ of size at most $k$.
    We define a function $f: V(G') \to \{0,1,2\}$ as follows: for every $v \in V(G)$ we set $f(v_1) = 0$ if $v \not\in X$ and $f(v_1) =2$ otherwise and $f(v_2) = 0$.
    The weight of $f$ is $2|X| \leq 2k$.
    Let us prove that $f$ is a RDF.
    For every $v \in V(G)$, there exists $w \in N[v]$ such that $w \in X$.
    As $w_1 \in N[v_1]$ and $w_1 \in N[v_2]$ and $w_1 \in V_2$, we deduce that $v_1$ and $v_2$ are dominated by $w_1$ such that $f(w_1) = 2$.
    Thus $f$ is a RDF of $G'$ of weight at most $2k$.

    Suppose that $G'$ has a RDF $f$ of weight at most $2k$.

    We define $X$ as the vertices of $V(G)$ such that $f(v_1) + f(v_2) \geq 2$.
    Let us prove that $X$ is dominating set of $G$.
    Let $u \in V(G) \not\in X$.
    Thus $f(v_1) = 0$ or $f(v_2) = 0$.
    If $f(v_1) = 0$, then $f(v_2) < 2$.
    In this case, there exists $w \in N(v)$ such that $f(w_1) = 2$ or $f(w_2) = 2$.
    We deduce that $w \in X$ is adjacent to $v$ in $G$.
    The case where $f(v_2) = 0$ is symmetric.
    We conclude that $X$ is a dominating set of $G$.

    Suppose by contradiction that $|X| > k$.
    Thus the weight of $f$ satisfies:
    \[
    w(f) = \sum_{v \in V(G)} (f(v_1) + f(v_2)) \geq \sum_{v \in X} (f(v_1) + f(v_2)) \geq \sum_{v \in X} 2 \geq 2|X| > 2k \;.
    \]
    This contradicts the fact that $f$ is of weight at most $2k$.
    We deduce that $|X| \leq k$. This completes the proof.
\end{proof}


\begin{corollary}
{\sc Roman Domination} parameterized by 2ccvd is \textsf{para-NP}-hard.    
\end{corollary}

\subsection{Complexity of {\sc Independent Roman Domination} in 2-clubs}

By a similar proof
to the previous one we 
obtain the following Theorem.

\begin{theorem} \label{theorem:IRD_2clubs}
{\sc Independent Roman Domination} is \WO-hard in $2$-clubs.
\end{theorem}


\begin{corollary}
    {\sc Independent Roman Domination} parameterized by 2ccvd is \textsf{para-NP}-hard.
\end{corollary}

\subsection{Complexity of {\sc Perfect Roman Domination} in $2$-clubs}

We reduce {\sc Perfect Roman Domination} from the following problem which has been proved to be \NP-complete \cite{Karp72}:

\prob{Exact Cover}
{A set of elements $U$, a set $S$ of subsets of $U$;}
{Is there an exact cover of $U$, that is a subset $S'$ of $S$ such that the sets in $S'$ are pairwise disjoint and the union of these sets is $U$?  }

\begin{theorem}
\label{theorem:PRD}
{\sc Perfect Roman Domination} is \NP-complete even restricted to $2$-clubs.
\end{theorem}

\begin{proof}

Let $U = \{u_1, \ldots,u_n\}$ and $S=\{s_1, \ldots, s_k\}$ be an instance of {\sc Exact Cover}. We define $w = 2k+3$ and a graph $G$ as follows:

\begin{itemize}

\item Add one vertex each set of $S$;
\item Add an independent set $U_i$ with $w+1$ vertices for every element $u_i$ of $U$;
\item Connect $s_i$ to $U_j$ for every $i,j$ such that $u_j \in s_i$;
\item Add an independent set $S_i$ with $w+1$ vertices for every set $s_i \in S$ and connect $s_i$ to $S_i$;
\item Add a vertex $r_i$ for every set $s_i \in S$ and connect $r_i$ to $s_i$ and $r_i$ to $S_i$;
\item Add an independent set $T_i$ with $w+1$ vertices for every set $s_i$ of $S$ and connect $r_i, S_i$ and $s_i$ to $T_i$;
\item Add a vertex $a$ connected to every $r_i, S_i, s_i$ and $U_j$;
\item Add a vertex $b$ connected to every $T_i, r_i, S_i, s_i$;
\item Add a vertex $c$ connected to every $T_i$ and $U_j$.
\end{itemize}

\begin{figure}[htb!]
\centering
	\begin{tikzpicture}[yscale=-1, scale = 1.7]

        \def\vertexIdSize{0.3}; 
        \def\weightSize{0.4}; 
        \def\edgeWidth{0.5}; 
        \def\vertexSize{0.2}; 
    
		\coordinate (v0) at (4.9, 2.8);
		\coordinate (v1) at (4.9, 3.85);
		\coordinate (v2) at (4.9, 4.9);
		\coordinate (v3) at (6.3, 2.8);
		\coordinate (v4) at (6.3, 3.85);
		\coordinate (v5) at (6.31, 4.9);
		\coordinate (v6) at (4.2, 2.63);
		\coordinate (v7) at (4.2, 2.98);
		\coordinate (v8) at (3.5, 2.8);
		\coordinate (v9) at (4.2, 3.68);
		\coordinate (v10) at (4.2, 4.03);
		\coordinate (v11) at (3.5, 3.85);
		\coordinate (v12) at (4.2, 4.73);
		\coordinate (v13) at (3.5, 4.9);
		\coordinate (v14) at (4.2, 5.08);
		\coordinate (v15) at (2.45, 3.85);
		\coordinate (v16) at (4.2, 1.75);
		\coordinate (v17) at (5.25, 5.6);

		\draw[line width = \edgeWidth, color = Neutral] (v1) --  (v3) ;
		\draw[line width = \edgeWidth, color = Neutral] (v1) --  (v5) ;
		\draw[line width = \edgeWidth, color = Neutral] (v4) --  (v0) ;
		\draw[line width = \edgeWidth, color = Neutral] (v0) --  (v3) ;
		\draw[line width = \edgeWidth, color = Neutral] (v4) --  (v2) ;
		\draw[line width = \edgeWidth, color = Neutral] (v2) --  (v5) ;
		\draw[line width = \edgeWidth, color = Neutral] (v7) --  (v0) ;
		\draw[line width = \edgeWidth, color = Neutral] (v0) --  (v6) ;
		\draw[line width = \edgeWidth, color = Neutral] (v6) --  (v7) ;
		\draw[line width = \edgeWidth, color = Neutral] (v8) --  (v6) ;
		\draw[line width = \edgeWidth, color = Neutral] (v8) --  (v7) ;
		\draw[line width = \edgeWidth, color = Neutral] (v8) --  (v0) ;
		\draw[line width = \edgeWidth, color = Neutral] (v9) --  (v10) ;
		\draw[line width = \edgeWidth, color = Neutral] (v10) --  (v11) ;
		\draw[line width = \edgeWidth, color = Neutral] (v11) --  (v1) ;
		\draw[line width = \edgeWidth, color = Neutral] (v11) --  (v9) ;
		\draw[line width = \edgeWidth, color = Neutral] (v9) --  (v1) ;
		\draw[line width = \edgeWidth, color = Neutral] (v1) --  (v10) ;
		\draw[line width = \edgeWidth, color = Neutral] (v2) .. controls (4.55, 4.81) ..  (v12) ;
		\draw[line width = \edgeWidth, color = Neutral] (v12) --  (v13) ;
		\draw[line width = \edgeWidth, color = Neutral] (v13) --  (v14) ;
		\draw[line width = \edgeWidth, color = Neutral] (v14) --  (v12) ;
		\draw[line width = \edgeWidth, color = Neutral] (v2) --  (v13) ;
		\draw[line width = \edgeWidth, color = Neutral] (v14) --  (v2) ;
		\draw[line width = \edgeWidth, color = Neutral] (v15) --  (v8) ;
		\draw[line width = \edgeWidth, color = Neutral] (v15) --  (v6) ;
		\draw[line width = \edgeWidth, color = Neutral] (v15) --  (v7) ;
		\draw[line width = \edgeWidth, color = Neutral] (v15) .. controls (3.83, 3.61) ..  (v0) ;
		\draw[line width = \edgeWidth, color = Neutral] (v15) --  (v12) ;
		\draw[line width = \edgeWidth, color = Neutral] (v15) .. controls (2.93, 5.26) ..  (v14) ;
		\draw[line width = \edgeWidth, color = Neutral] (v15) --  (v13) ;
		\draw[line width = \edgeWidth, color = Neutral] (v15) .. controls (4.37, 4.36) ..  (v2) ;
		\draw[line width = \edgeWidth, color = Neutral] (v8) --  (v16) ;
		\draw[line width = \edgeWidth, color = Neutral] (v16) --  (v3) ;
		\draw[line width = \edgeWidth, color = Neutral] (v17) --  (v14) ;
		\draw[line width = \edgeWidth, color = Neutral] (v17) --  (v12) ;
		\draw[line width = \edgeWidth, color = Neutral] (v17) --  (v2) ;
		\draw[line width = \edgeWidth, color = Neutral] (v17) --  (v5) ;

		\node[scale = \vertexSize, label={[text=white, scale=\vertexIdSize]center: }, nodes={Neutral}{}{}{}] at  (v0)  {};
		\node[label={[scale=\weightSize]below: s1}] at (v0) {};
		\node[scale = \vertexSize, label={[text=white, scale=\vertexIdSize]center: }, nodes={Neutral}{}{}{}] at  (v1)  {};
		\node[label={[scale=\weightSize]below: s2}] at (v1) {};
		\node[scale = \vertexSize, label={[text=white, scale=\vertexIdSize]center: }, nodes={Neutral}{}{}{}] at  (v2)  {};
		\node[label={[scale=\weightSize]below: s3}] at (v2) {};
		\node[scale = \vertexSize, label={[text=white, scale=\vertexIdSize]center: }, nodes={Blue}{}{}{}] at  (v3)  {};
		\node[label={[scale=\weightSize]right: $U_1$}] at (v3) {};
		\node[scale = \vertexSize, label={[text=white, scale=\vertexIdSize]center: }, nodes={Blue}{}{}{}] at  (v4)  {};
		\node[label={[scale=\weightSize]right: $U_2$}] at (v4) {};
		\node[scale = \vertexSize, label={[text=white, scale=\vertexIdSize]center: }, nodes={Blue}{}{}{}] at  (v5)  {};
		\node[label={[scale=\weightSize]right: $U_3$}] at (v5) {};
		\node[scale = \vertexSize, label={[text=white, scale=\vertexIdSize]center: }, nodes={Blue}{}{}{}] at  (v6)  {};
		\node[label={[scale=\weightSize]above: $S_1$}] at (v6) {};
		\node[scale = \vertexSize, label={[text=white, scale=\vertexIdSize]center: }, nodes={Neutral}{}{}{}] at  (v7)  {};
		\node[label={[scale=\weightSize]below: $r_1$}] at (v7) {};
		\node[scale = \vertexSize, label={[text=white, scale=\vertexIdSize]center: }, nodes={Blue}{}{}{}] at  (v8)  {};
		\node[label={[scale=\weightSize]left: $T_1$}] at (v8) {};
		\node[scale = \vertexSize, label={[text=white, scale=\vertexIdSize]center: }, nodes={Blue}{}{}{}] at  (v9)  {};
		\node[label={[scale=\weightSize]above: $S_2$}] at (v9) {};
		\node[scale = \vertexSize, label={[text=white, scale=\vertexIdSize]center: }, nodes={Neutral}{}{}{}] at  (v10)  {};
		\node[label={[scale=\weightSize]below: $r_2$}] at (v10) {};
		\node[scale = \vertexSize, label={[text=white, scale=\vertexIdSize]center: }, nodes={Blue}{}{}{}] at  (v11)  {};
		\node[label={[scale=\weightSize]left: $T_2$}] at (v11) {};
		\node[scale = \vertexSize, label={[text=white, scale=\vertexIdSize]center: }, nodes={Blue}{}{}{}] at  (v12)  {};
		\node[label={[scale=\weightSize]above: $S_3$}] at (v12) {};
		\node[scale = \vertexSize, label={[text=white, scale=\vertexIdSize]center: }, nodes={Blue}{}{}{}] at  (v13)  {};
		\node[label={[scale=\weightSize]left: $T_3$}] at (v13) {};
		\node[scale = \vertexSize, label={[text=white, scale=\vertexIdSize]center: }, nodes={Neutral}{}{}{}] at  (v14)  {};
		\node[label={[scale=\weightSize]below: $r_3$}] at (v14) {};
		\node[scale = \vertexSize, label={[text=white, scale=\vertexIdSize]center: }, nodes={Neutral}{}{}{}] at  (v15)  {};
		\node[label={[scale=\weightSize]left: $b$}] at (v15) {};
		\node[scale = \vertexSize, label={[text=white, scale=\vertexIdSize]center: }, nodes={Neutral}{}{}{}] at  (v16)  {};
		\node[label={[scale=\weightSize]below: $c$}] at (v16) {};
		\node[scale = \vertexSize, label={[text=white, scale=\vertexIdSize]center: }, nodes={Neutral}{}{}{}] at  (v17)  {};
		\node[label={[scale=\weightSize]below: $a$}] at (v17) {};

	\end{tikzpicture}

\caption{Example of a graph obtained by the construction of Theorem~\ref{theorem:PRD} for the {\sc Exact Cover} instance $(U=\{u_1, u_2, u_3\}, S=\{s_1=\{u_1, u_2\}, s_2 =\{u_1, u_3\}, s_3=\{u_2, u_3\})$. For readability, independent sets $U_j$, $S_i$ and $T_i$ are represented by a single blue vertex and some edges incident to $a,b$ and $c$ are not described.}
        \label{fig:construction_PRD}
    \end{figure}
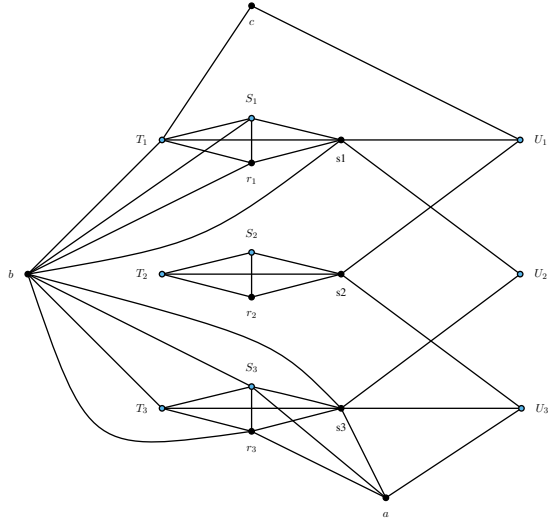

The graph $G$ is of diameter at most $2$.
We prove this by giving a common neighbor for every pair of vertices of $G$ in Table~\ref{table:prd_common}.

\begin{table}[htb!] 
        \centering
        \begin{tabular}{c|cccccccc}
             &  $U_j$ & $s_i$ & $S_i$ & $r_i$ & $T_i$ & $a$ & $b$ & $c$ \\ \hline
             $U_{j'}$ & $a$ & $a$ & $a$ & $a$ & $c$ & $a$ & $s_i$ & $c$  \\
             $s_{i'}$ &  & $a$ & $a$ & $a$ & $b$ & $a$ & $b$ & $u_j \in s_{i'}$ \\
             $S_{i'}$ & & & $a$ & $a$ & $b$ & $a$ & $b$ & $T_{i'}$ \\
             $r_{i'}$ &  & & & $a$ & $b$ & $a$ & $b$ & $T_{i'}$ \\
             $T_i$ & & & & & $c$ & $r_i$ & $b$ & $c$ \\
             $a$ & & & & & & $a$ & $s_1$ & $U_1$ \\
             $b$ & & & & & & & $b$ & $T_1$ \\
             $c$ & & & & & & & & $c$
        \end{tabular}
        \caption{Common neighbors of $G$ for each type of vertices, proving that $G'$ is of diameter $2$. Sub-diagonal cases are symmetric.}
\label{table:prd_common}
    \end{table}

Let us prove that $(U,S)$ has a solution of size at most $k$ if and only if $G$ has a PRDF of size at most $w=2k+3$.
Suppose $(U,S)$ has a solution $X$.
We define the following function $f: V(G) \to \{0,1,2\}$.
For every $i \in X$, $f(s_i) = 2$.
For every $i \not\in X$, $f(r_i) = 2$.
Furthermore we set $f(a) = f(b) = f(c) = 1$.
For every other vertex $v$, we set $f(v) = 0$.
The weight of $f$ is $2k+3$.
Let us prove that $f$ is a RDF.
For every $i$, each vertex of $S_i$ and $T_i$ is only dominated by $s_i$ and $r_i$.
For every $i \in X$, $r_i$ is only dominated by $s_i$.
For every $i \not\in X$, $s_i$ is only dominated by $r_i$.
For every $j$, there exists $i$ such that $u_j \in s_i$, thus each vertex of $U_j$ is only dominated by $s_i$.
We deduce that $f$ is a PRDF.

Suppose $G$ has a PRDF $f$ of weight at most $w=2k+3$ and let us prove that $(U,S)$ has an exact cover. Suppose by contradiction that $f(b) = 2$.
Then $f(s_i) < 2$, $f(S_i) < 2$, $f(r_i) < 2$ for every $i$, otherwise one of this vertex $v$ would have value $f(v) = 2$ and then all vertices of $T_i$ would be adjacent to $v$ and to $b$, so all vertices would have value at least $1$.
As $|T_i| > w$, it would contradict that $f$ is of weight at most $w$.
In the same way, we deduce that $f(a) < 2$ (otherwise all vertices of $S_1$ would be adjacent to $a$ and $b$ which have value $2$).
We also deduce that $f(c) < 2$.
Thus for every $j$, every vertex of $U_j$ is not adjacent to a vertex valued $2$ and is therefore of value at least $1$, a contradiction.
We deduce that $f(b) \leq 1$.
    
In the same way, let us prove that $f(c) \leq 1$.
Suppose by contradiction that $f(c) =2$.
Then $f(a) < 2$ (because $c$ and $a$ are connected to every $U_j$ ).
Furthermore every vertex $r_i, S_i, s_i$ is of value at most $1$ because they are all connected to $T_i$.
As $S_i$ must contain a vertex valued $0$ (otherwise the weight of $f$ would be greater than $w+1$), this vertex can only be covered by a $2$ situated in $T_i$.
    We deduce that every $T_i$ has a $2$.
    As $b$ is connected to every $T_i$, it cannot be assigned to $0$ and is thus assigned $1$.
    It is not possible that $f(a) = 0$, otherwise there would exist a vertex $v$ in $U_j$ assigned to $2$.
    In this case the weight is at least $f(c) + f(b) + f(a) + f(v) + \sum_{i} \sum_{v \in T_i} f(v) \geq 2 + 1 +2 + 2k = 4 + 2k$.
    This contradicts the upper bound of $w(f)$.
    We deduce that $f(b) \leq 1$.

    As $f(b) \leq 1$ and $f(c) \leq 1$ and as every $T_i$ must contain a $0$, then for every $i$, $r_i$, $s_i$ or a vertex of $S_i$ is assigned to $2$.
    Thus the weight of $f$ is at least $f(a) + f(b) + f(c) + \sum_{i} f(r_i) + f(s_i) + f(S_i) \geq 3 + 2k$.
    We deduce that every other vertices are assigned to $0$, in particular the vertices of $U_j$ for every $j$.

    We define $X$ as $\{ i: f(s_i) = 2\}$.
    Let us prove that $X$ is an exact cover of $U$.
    Let $u_j \in U_j$.
    As $f(u_j) = 0$ and as $u_j$ is connected to $T_i$ (which are assigned to $0$), to $a$ and $c$ (which are assigned to $1$) and to some vertices $s_i$, thus there exists $s_i$ with $f(s_i) = 2$ which is connected to $u_j$.
    As $f$ is a perfect roman domination, we deduce that there exists exactly one vertex $s_i$ with $f(s_i) = 2$ which is connected to $u_j$.
    We conclude that $X$ is an exact cover of $U$.

    Finally, $(U,S)$ has an exact cover if and only if $G$ has a PRDF of weight at most $2k+3$.
    We conclude that {\sc Perfect Roman Domination} is \NP-complete in 2-clubs.
\end{proof}


\begin{corollary}
    {\sc Perfect Roman Domination}-{\sc 2CCVD} is \textsf{para-NP}-hard.
\end{corollary}

\section{Fixed-Parameter Sub-Exponential Algorithms}

Despite its W[2]-hardness, the {\sc Dominating Set} problem is solvable in sub-exponential time in 2-clubs. In fact, the minimum size of a dominating set in a 2-club is in $\mathcal{O}(\sqrt{n\log{n}})$. This property, which is based on the fact that the neighborhood of any vertex in a 2-club is a dominating set, was used by Mertzios and Spirakis to solve the {\sc 3-Coloring} problem in sub-exponential time in 2-clubs \cite{MertziosS16}. To see this: once a dominating set $D$ of size $c\sqrt{n\log{n}}$ is found, a simple brute-force {\sc 3-Coloring} 
algorithm consists of enumerating all the possible $\mathcal{O}(3^{|D|})$ 3-colorings of the elements of $D$, which restricts the number of colors of each other vertex (in $G-D$) to at most two. Completing the 3-coloring, or reporting a No-instance, consists of solving an instance of {\sc 2-List Coloring} in polynomial-time via a simple reduction to 2-SAT. More recently, this result was further improved by Debski \textit{et al.} \cite{DebskiPR22} where {\sc 3-Coloring} and {\sc 3-List Coloring} were shown to be solvable in $2^{\mathcal{O}(n^{\frac{1}{3}}\log^2{n})}$.

Let $(G,k)$ be an instance of {\sc 3-Coloring} parameterized by 2ccvd, and let $S$ be such that $G-S$ is a disjoint union of 2-clubs, with $|S|\leq k$. In this case, a simple {\sc 3-Coloring} algorithm consists of enumerating all the $\mathcal{O}(3^k)$ possible colorings of $G[S]$ and then, in each 2-club $C$ of $G-S$ we use the algorithm described in \cite{DebskiPR22} to complete the 3-coloring, if possible. 
Overall, we obtain the following.


\begin{theorem}
{\sc 3-Coloring}, parameterized by 2ccvd, is solvable in
$\mathcal{O}(3^k 2^{cn^{1/3}\log^2(n)})$ time.
\end{theorem}

A similar result can be obtained 
for the 3ccvd parameter, due to a sub-exponential algorithm for {\sc 3-List Coloring} that runs in sub-exponential time on diameter-3 graphs \cite{DebskiPR22}.

In general, and because of the Exponential-Time Hypothesis (ETH), solvability in sub-exponential time when a parameter is fixed (or for each constant value of the parameter) can be interesting by itself, as an analogy to solvability in polynomial-time. However, it is not clear yet whether a complexity hierarchy exists that is analogous to the W-hierarchy. 
On the other hand, one can show that, modulo the ETH, some classical problems are not solvable by fixed-parameter sub-exponential algorithms with respect to 2ccvd.
Notable examples include {\sc Vertex Cover} and {\sc Independent Set}. To see this, let $(G,k)$ be an arbitrary instance of {\sc Vertex Cover}. If the diameter of $G$ is greater than two, we construct an equivalent diameter-two instance $(G',k+1)$ simply by adding a vertex $s$ that is adjacent to all vertices of $G$ (in $G'$). The size of the vertex cover increases by 1 simply because $G$ is not a complete graph (being of diameter greater than $2$). This shows that solving {\sc Vertex Cover} by a sub-exponential algorithm on 2-clubs results in solving it in sub-exponential time in general, which is impossible unless the ETH fails \cite{impagliazzo2001problems}. The same reduction applies to {\sc Independent Set} but with the same value of 
$k$.


The above simple reduction poses more questions. Which problems are solvable via fixed-parameter subexponential algorithms when parameterized by 2ccvd? We showed in another manuscript (in preparation) that the size of a minimum connected dominating set is bounded above by $2\sqrt{n}\log{n}$ in 2-clubs, which leads to solving {\sc Connected Dominating Set} via a sub-exponential algorithm in 2-clubs. The same is known for {\sc Dominating Set}. Is any of the two problems solvable via a fixed-parameter subexponential algorithm w.r.t. 2ccvd? We believe this is an interesting open problem.

Another interesting scenario where several problems are solvable via fixed-parameter sub-exponential algorithms is when diameter-two is replaced with planarity.
This follows for the fact that many \NP-hard problems are solvable in sub-exponential time on planar graphs due to the planar separator theorem \cite{LiptonT80}. The same applies to several other parameters including distance to $P_t$-free and Broom-free graphs \cite{BacsoLMPTL19}.

Motivated by the above, one can define the class of problems solvable via fixed-parameter sub-exponential algorithms, which can possibly be named \textit{Fixed-Parameter Sub-exponential} or \FPSUB. 
Obviously $\FPT \subseteq \FPSUB$, but what about higher parameterized complexity classes? We already observed that {\sc Dominating Set} is \WO-hard w.r.t 2ccvd, so it does not seem possible to place \FPSUB\ in the W-hierarchy. We believe this classification is worth exploring.




\bibliography{references}

\end{document}